\documentclass[nopreprintline]{elsarticle}

\usepackage{booktabs} %
\usepackage{url}
\usepackage[utf8]{inputenc}
\usepackage[ruled]{algorithm2e} %
\usepackage{xargs}
\usepackage[inline]{enumitem}
\usepackage{amsthm}
\usepackage{amsmath}
\usepackage{amssymb}
\usepackage{macros}

\newcommand{\supp}{\mathit{support}}
\newcommand{\weight}{\mathit{weight}}
\newcommand{\fhw}{\mathit{fhw}}
\newcommand{\ghw}{\mathit{ghw}}
\newcommand{\np}{\textsf{NP}}

\newcommand{\OPT}{\mathit{OPT}}

\newcommand{\nop}[1]{}
\newcommand{\abs}[1]{\lvert #1 \rvert}

\newtheorem{theorem}{Theorem}
\newtheorem{proposition}[theorem]{Proposition}
\newtheorem{corollary}[theorem]{Corollary}
\newtheorem{lemma}[theorem]{Lemma}

\theoremstyle{definition}
\newdefinition{definition}[theorem]{Definition}
\newdefinition{remark}{Remark}

\newtheorem{example}{Example}

\newtheorem{question}{Open Question}
\newtheorem{conjecture}{Conjecture}

\title{Fractional Covers of Hypergraphs with Bounded Multi-Intersection\tnoteref{t1}}

\begin{document}

\begin{frontmatter}
  \tnotetext[t1]{An extended abstract of this article was presented at the 45th International Symposium on Mathematical Foundations of Computer Science (MFCS 2020)~\cite{DBLP:conf/mfcs/GottlobLPR20}.}
  
  \author[ox]{Georg Gottlob}
  \ead{georg.gottlob@cs.ox.ac.uk}

  \author[ox]{Matthias Lanzinger\corref{cor1}}
  \ead{matthias.lanzinger@cs.ox.ac.uk}

  \author[rp]{Reinhard Pichler}
  \ead{pichler@dbai.tuwien.ac.at}

  \author[ir]{Igor Razgon}
  \ead{igor@dcs.bbk.ac.uk}

  \cortext[cor1]{Corresponding author}
  \address[ox]{University of Oxford, Wolfson Building, Parks Road, Oxford OX1 3AQ, United Kingdom}
  \address[rp]{TU Wien, Favoritenstrasse 9, 1040 Wien, Austria}
  \address[ir]{Birkbeck University of London, Malet Street, London WC1E 7HX, United Kingdom}


  \begin{abstract}
Fractional (hyper-)graph theory is concerned with the specific problems that arise when fractional analogues of otherwise integer-valued (hyper-)graph invariants are considered. The focus of this paper is on fractional edge covers of hypergraphs. Our main technical result generalizes and unifies previous conditions under which the size of the support of fractional edge covers is bounded independently of the size of the hypergraph itself. We show how this combinatorial result can be used to extend previous tractability results for checking if the fractional hypertree width of a given hypergraph is $\leq k$ for some constant $k$. Moreover, we show a dual version of our main result for fractional hitting sets.
  \end{abstract}

\end{frontmatter}

\section{Introduction}
\label{sec:intro}

Fractional (hyper-)graph theory \cite{fractionalGraphTheory} has evolved into a mature discipline 
in graph theory  -- building upon early research efforts that date back to the 1970s
\cite{bergeFractionalGraphTheory}. The crucial observation \revision{motivating} this field is that many integer-valued 
(hyper-)graph invariants have a meaningful fractional analogue. 
Frequently,
the integer-valued invariants are defined in terms of an integer linear program (ILP) and the fractional analogue is 
obtained by the fractional relaxation. 
Examples of problems which have been studied in 
fractional (hyper-)graph theory 
comprise matching problems, coloring problems, covering problem and many more. 

Covering problems come in two principal flavors, namely \emph{edge covers} \revision{(also referred to as set covers)} and \emph{hitting sets} (also referred to as vertex covers). 
We shall concentrate on edge covers in the first place, and afterwards show how our results 
translate to hitting sets.
{\em Fractional edge covers\/} have attracted a lot of attention in recent times. 
On the one hand, this is due to a deep connection between information theory and database theory. 
Indeed, the famous ``AGM bound'' -- named after \revision{Atserias, Grohe, and Marx}~\cite{DBLP:journals/siamcomp/AtseriasGM13}~-- 
establishes a tight upper bound on the number of result tuples of relational joins in terms of fractional edge covers.
On the other hand, fractional hypertree width ($\fhw$) 
is to date the most general width-notion that allows one to define
tractable fragments of solving Constraint Satisfaction Problems (CSPs), 
answering Conjunctive Queries (CQs), 
and solving the Homomorphism Problem \cite{2014grohemarx}. 
The fractional hypertree width of a hypergraph
is defined in terms of the size of fractional edge covers of the bags in a tree decomposition.

Fractional (hyper-)graph invariants 
give rise to new challenges that do not  exist in the integral case. 
Intuitively, if a fractional (hyper-)graph invariant is obtained by the relaxation of a linear program (LP), one would expect things to become
easier, since we move from the intractable problem of ILPs to the tractable problem of LPs.
However, also the opposite may happen, namely that the fractional relaxation introduces complications not present
in the integral case.  To illustrate such an effect, we first recall some basic definitions.

\begin{definition}
\label{def:edgeCover}
A \emph{hypergraph} $H$ is a tuple $H = (V, E)$, consisting of a set of vertices 
$V$ and a set of hyperedges (or simply edges), which are non-empty subsets of $V$.
Let $\gamma$ be a function of the form $\gamma \colon E \rightarrow \revision{\mathbb{R}^+}$\revision{, i.e., mapping edges to the nonnegative reals,} 
Then the set of vertices \emph{covered} by $\gamma$ is defined as 
$B(\gamma) = \{ v\in V \mid \sum_{e\in E, v\in e} \gamma(e) \geq 1 \}$.
Intuitively, $\gamma$ assigns weights to the edges and a vertex $v$ is covered if 
the total weight of the edges containing $v$ is at least 1.

A  {\em  fractional edge cover\/} of $H$ is a function $\gamma$ with $V  \subseteq B(\gamma)$. 
An  {\em  integral edge cover\/} is obtained by restricting the function values of $\gamma$ to $\{0,1\}$.
The \emph{support} of $\gamma$ is defined as $\supp(\gamma) = \{e \in E \, \mid \, \gamma(e) \neq 0\}$. 
The \emph{weight} of $\gamma$ is defined as $\weight(\gamma) =  \sum_{e\in E} \gamma(e)$. 
The minimum weight of a fractional (resp.\ integral) edge cover of a hypergraph $H$ is referred to as the
fractional (resp.\ integral) \emph{edge cover number} of $H$.
\end{definition}
The following example adapted from \cite{DBLP:conf/pods/FischlGP18} illustrates which complications may arise if we 
move from the integral to the fractional case.

\begin{example}
\label{ex:LongEdge}
Consider the family $(H_n)_{n \geq 2}$ of hypergraphs with $H_n =(V_n,E_n)$ defined as 

\smallskip
$V_n = \{v_0, v_1, \dots, v_n\}$

$E_n= \{ e_0, e_1, \dots, e_n\}$ with 
$e_0 = \{v_1, \dots, v_n\}$ and $e_i = \{v_0,v_i\}$ for $i \in \{1, \dots, n\}$.

\smallskip

\noindent
The integral edge cover number of each $H_n$ is 2 and an optimal integral edge cover can be obtained, e.g., 
by setting $\gamma_n(e_0) = \gamma_n(e_1) = 1$ and $\gamma_n(e) = 0$ for all other edges. 
In contrast, the fractional edge cover number is $2 - \frac{1}{n}$ and the 
unique optimal fractional edge cover is $\gamma'_n$ with 
$\gamma'_n(e_0) = 1 - \frac{1}{n}$ 
and $\gamma'_n(e_i) = \frac{1}{n}$ for each $i \in \{1, \dots, n\}$.
For the support of these two covers, 
we have  $\abs{\supp(\gamma_n)}= 2$ and $\abs{\supp(\gamma'_n)}= n+1$. 
Hence,
the support of the optimal edge covers is bounded in the integral case but unbounded in the
fractional case.
\hfill
$\diamond$
\end{example}

As mentioned above, 
fractional hypertree width ($\fhw$) 
is to date the most general width-notion that allows one to define
tractable fragments of classical NP-complete problems, such as CSP solving and CQ answering. 
However, recognizing if a given hypergraph $H$ has $\fhw (H) \leq k$ for fixed $k \geq 2$ is 
itself an NP-complete problem \cite{DBLP:conf/pods/FischlGP18}\revision{, i.e., in the terminology of parameterised complexity, the problem is paraNP-hard}. 
It has recently been shown that the problem of checking $\fhw(H)\leq k$ becomes tractable
if we can efficiently enumerate the fractional edge covers of \revision{weight} $\leq k$ 
\cite{JACM2021}. This fact can be exploited to show that,
for classes of hypergraphs  
with bounded rank (i.e., max.\ size of edges), bounded degree
(i.e., max.\ number of edges containing a particular vertex), or bounded intersection
(i.e., max.\ number of vertices in the intersection of two edges), checking  
$\fhw(H)\leq k$ becomes tractable.
The size of the support has been recently 
\cite{JACM2021}
identified as a crucial parameter 
for the efficient enumeration of fractional edge covers of weight $\leq k$ for given $k \geq 1$.

The {\em overarching goal\/} of this work is to further extend and provide a uniform view 
of previously known 
structural properties of hypergraphs that guarantee a bound on the size of the support 
of fractional edge covers of a given weight. 
In particular, when looking at  Example~\ref{ex:LongEdge}, we want to avoid the
situation that the support of fractional edge covers with constantly bounded weight increases with the size of the hypergraph. 
Our {\em main combinatorial result\/} (Theorem~\ref{maintheor}) will be that the size of the support of a fractional edge cover 
does not depend on the number of vertices or edges of a hypergraph but
instead only on the weight of the cover as well as the structure of its edge intersections.

Formally, the structure of the edge intersections is captured by the so-called 
{\em Bounded Multi-Intersection Property (BMIP)\/} \cite{DBLP:conf/pods/FischlGP18}: 
a class ${\cal C}$ of hypergraphs has this property, if in every hypergraph $H \in {\cal C}$, 
the intersection of $c$ edges of $H$ has at most $d$ elements, 
for constants $c \geq 2$ and $d \geq 0$.
The BMIP thus generalizes all of the above mentioned hypergraph properties that ensure 
bounded support of fractional edge covers of given weight and, hence, also 
guarantee tractability of checking $\fhw(H)\leq k$, namely  
bounded rank, bounded degree, and bounded intersection. 
Moreover, when considering the incidence graph $G$ of $H$, 
the BMIP corresponds to $G$ not containing large complete bipartite graphs.
A notable result in the area of parameterized complexity \cite{DBLP:journals/talg/PhilipRS12}
is the polynomial kernelizability of the Dominating Set Problem for graphs without $K_{c,d}$, i.e., without the complete bipartite graph on $c$ and $d$ edges.
A minor tweaking of the results yields a polynomial kernelization for the Set Cover Problem
if the corresponding incidence graph does not contain $K_{c,d}$.
Our result thus reveals an interesting connection: 
it shows that a condition that enables efficient solving of the Set Cover problem
also enables efficient checking of bounded fractional hypertree width. 

\nop{
The Set Cover problem can be considered a special case of the 
Hypertree Width problem. In fact, the reduction to Set Cover was used
to demonstrate the W-hardness of computing Hypertree Width.
}

In summary, the main results of this paper are as follows: 
\begin{itemize}%
\item First, we show that the size of the support of a fractional edge cover 
only depends on the weight of the cover and of the structure of its edge intersections (Theorem~\ref{maintheor}). 
More specifically, if 
the intersection of $c$ edges of a hypergraph $H$ has at most $d$ elements, 
and $H$ has a fractional edge cover of weight $\leq k$, 
then $H$ also has a fractional edge cover of weight $\leq k$ with a support whose size 
only depends on $c,d$, and $k$.
\item As an important consequence of this result, we 
show that the problem of checking if a given hypergraph $H$ has $\fhw(H) \leq k$
is tractable for hypergraph classes satisfying the BMIP (Theorem~\ref{mainres}).
In particular, BMIP generalizes all previously known hypergraph classes with tractable 
$\fhw$-checking,  namely bounded rank, bounded degree, and bounded intersection.

\item We transfer our results on fractional edge covers to fractional hitting sets, 
where we again vastly generalize previously known hypergraph classes (such as 
hypergraphs of bounded 
rank \cite{furedi1988}) that guarantee
a bound on the size of the support of fractional~hitting~sets (Theorem~\ref{thm:vcsupp}). 
\end{itemize}

\noindent
The paper is organized as follows: after recalling some basic notions and results in Section~\ref{sec:prelim}, 
we will present our main technical result on fractional edge covers in Section~\ref{sec:tech}. 
The detailed proof of a crucial lemma is separated in Section~\ref{sec:boundext}.
In Section~\ref{sec:apps}, we apply our result on the bounded support of fractional edge covers to 
fractional hypertree width and extend our main combinatorial result to fractional hitting sets. Finally, in Section~\ref{sec:conclusion}, we 
summarize our results and give an extensive overview of interesting open questions in this area of research.

\section{Preliminaries}
\label{sec:prelim}

\paragraph{ Some general notation}
It is convenient to use the following short-hand notation for various kinds of sets:
we write $[n]$ for the set $\{1,\dots, n\}$ of natural numbers. 
Let $S$ be a set of sets. Then we write $\bigcap S$  and $\bigcup S$ for the intersection
and union, respectively, of the sets in $S$, i.e., 
 $\bigcap S = \{x \, \mid \, x \in s$ for all $s \in S\}$ and
 $\bigcup S = \{x \, \mid \, x \in s$ for some $s \in S\}$.

\paragraph{Hypergraphs}
We recall some basic notions on hypergraphs. 
We have already introduced 
in Section~\ref{sec:intro}
hypergraphs as pairs $(V,E)$ consisting of a set $V$ of vertices and 
a set $E$ of edges. W.l.o.g., we 
assume throughout this paper that a hypergraph \revision{contains no isolated vertices (i.e., vertices
that do not occur in any edge), no pair of vertices that are incident to the exact same set of edges, and no empty edges. 
Such hypergraphs are typically referred to as \emph{reduced} hypergraphs}.
Given a hypergraph $H = (V,E)$, the \emph{dual hypergraph}
$H^d  = (W,F)$ 
is defined as $W = E$ and $F = \{ \{e \in E \mid v \in e\} \mid v \in V\}$. 

The {\em incidence graph\/} of a hypergraph $H = (V,E)$ is a bipartite graph $(W,F)$ with 
$W = V \cup E$, such that, for every  $v\in V$ and $e \in E$, there is an edge \revision{$\{v,e\}$} in $F$ iff 
$v \in e$. Note that a hypergraph $H$ and its dual hypergraph $H^d$ have
the same incidence graph. 

In this work, we are particularly interested in the structure of the edge intersections of a hypergraph. 
To this end, recall the notion of $(c,d)$-hypergraphs for integers 
$c \geq 2$ and $d \geq 0$
from \cite{JACM2021}: 
$H = (V,E)$ is a $(c,d)$-hypergraph if the intersection of any $c$ distinct edges in $E$ 
has at most $d$ elements, i.e., for every subset $E' \subseteq E$ with
$|E'| = c$, we have $|\bigcap E'| \leq d$.  
A class ${\cal C}$ of hypergraphs is said to satisfy the {\em bounded multi-intersection property (BMIP)}
\cite{DBLP:conf/pods/FischlGP18}, if there exist
$c \geq 2$ and $d \geq 0$, such that every $H$ in ${\cal C}$ is a $(c,d)$-hypergraph. As a special
case studied in \cite{fischl2019hyperbench,DBLP:conf/pods/FischlGP18}, a 
class ${\cal C}$ of hypergraphs is said to satisfy the {\em bounded intersection property (BIP)},
if there exists $d \geq 0$, such that every $H$ in ${\cal C}$ is a $(2,d)$-hypergraph. Hypergraphs with degree bounded by some constant 
$c \geq 1$ are $(c+1,0)$-hypergraphs. Moreover, bounded rank is clearly a special case of bounded intersection, that is, if the size of each
hyperedge is bounded a constant $d$, also the intersection of any 
two hyperedges is of course bounded by $d$.

We now recall tree decompositions, which form the basis of various notions 
of width.
A tuple $(T, (B_u)_{u \in T})$ is a \emph{tree  decomposition (TD)\/} 
of a hypergraph $H =(V,E)$, if $T$ is a tree, 
every $B_u$ is a subset of $V$ and the following two conditions are satisfied:
\begin{enumerate}[label=(\arabic*)]
\item 
 For every edge $e \in E$, there is a node $u$ in $T$,  such that  $e \subseteq B_u$, and
\item for every vertex $v \in V$,  $\{u \in T \mid v \in B_u\}$ is connected in $T$.
\end{enumerate}
Note that, by slight abuse of notation, we write $u \in T$ to express that $u$ is a node in $T$.

For a function $f\colon 2^{V} \to \mathbb{R}^+$, the
\emph{$f$-width} of a TD $(T, (B_u)_{u \in T})$ 
is defined as 
$\sup\{f(B_u) \mid u \in T\}$ and the $f$-width of a hypergraph is the
minimal $f$-width over all its TDs.

An edge weight function is a function $\gamma: E \to \mathbb{R}^+$.  
We call $\gamma$ a \emph{fractional edge cover} of a set $X \subseteq V$
by edges in $E$, 
if for every $v \in X$, we have $\sum_{ \{e \, \mid \, v \in e\}} \gamma(e) \geq 1$.
The weight of a fractional edge cover is defined as 
$\weight(\gamma) = \sum_{e \in E} \gamma(e)$.
For $X \subseteq V$, we write 
$\rho_H^*(X)$ to denote the minimal weight over all fractional edge covers of $X$. 
\revision{With respect to some edge weight function $\gamma$, we will say that the \emph{weight of a vertex} is the sum of all the weights on edges that contain $v$.}
The {\em fractional hypertree width (fhw)\/} of a hypergraph $H$, denoted $\fhw(H)$,
is then defined as the $f$-width for $f = \rho_H^*$. Likewise, the $\fhw$ of a TD of $H$ 
is its $\rho_H^*$-width.

The following technical lemma for weight-functions in $(c,d)$-hypergraphs will be important.

\newcommand{\lemmaSmallEdges}{
There is a function $f(c,d,k)$ with the following property:
let $H$ be a $(c,d)$-hypergraph and let $\gamma$
be an edge  weight function with $\weight(\gamma) \leq k$. 
Moreover, let $0< \epsilon \leq 1$ be dependant on $c, d, k$ 
and assume that, for each $e \in E$, $\gamma(e) \leq \frac{\epsilon}{2c}$. 
Let $B^{\epsilon}(\gamma)$  be the set of all vertices of weight
at least $\epsilon$. Then $|B^{\epsilon}(\gamma)| \leq f(c,d,k)$ holds. 
}

\begin{lemma} \label{smalledges}
\lemmaSmallEdges
\end{lemma}

The intuition of this lemma is as follows: suppose that we put rather little weight on each edge 
(namely $\leq \frac{\epsilon}{2c}$). Then, for a vertex $v$ to be in $B^{\epsilon}(\gamma)$
(that is, to receive total weight $\geq \epsilon$ from all the edges containing $v$), 
$v$ must be in the intersection of quite a big number of edges (namely $\geq 2c$ such edges). 
However,  in a $(c,d)$-hypergraph, 
the intersection of $c$ or more edges contains at most $d$ vertices.
So, intuitively, the ``contribution'' of each edge to covering a particular vertex is limited. 
Hence, a weight function $\gamma$ with $\weight(\gamma) \leq k$ can 
put the desired weight $\geq \epsilon$ only to a limited number of vertices, where this limit
depends on $c,d$, and $k$. 

\begin{proof}[Proof of Lemma \ref{smalledges}]
The proof is based on the following claim 
(which is Lemma~7.2 in~\cite{JACM2021}).

\medskip

\noindent
{\sf Claim A.}
Fix an integer $c\geq 1$.
  \revision{Let $X= (x_1, \dots, x_n )$ be a sequence} of positive numbers $\leq \delta$
  and fix $w$ such that $\sum_{j=1}^{n}x_j \geq w \geq \delta c$.
  Then we have $\sum x_{i_1}\cdot x_{i_2} \cdot \,\cdots\, \cdot x_{i_c} \geq (w-\delta c)^c$,
  where the sum is over all $c$-tuples $(i_1, \dots, i_c)$ of \emph{distinct} integers
  from $[n]$.

\medskip

We proceed with a counting argument.
Imagine a bipartite graph $G = (B^{\epsilon}(\gamma), T, E(G))$ where $T$ is the set of all $c$-tuples 
of distinct edges from $H$.
In $G$, there is an edge from $v \in B^{\epsilon}(\gamma)$ to $(e_1, \dots, e_c) \in T$ iff $v$ is in $e_1 \cap \cdots \cap e_c$.
Furthermore, we assign weight $\prod_{j=1}^c \gamma(e_j)$ to every edge in $E(G)$ 
incident to a tuple $(e_1, \dots, e_c) \in T$.
To avoid confusion, in this proof, we write $E(G)$ and $E(H)$ to refer to the set of edges in the graph $G$ 
and in the hypergraph $H$, respectively.

We now count the total weight in $G$ from both sides. First observe that on the $T$ side, we have degree at most $d$ 
because $H$ is a $(c,d)$-hypergraph. 
Therefore, the total weight in $G$ is at most $d \cdot \sum_{(e_1, \dots, e_c) \in T} \prod_{j=1}^c \gamma(e_j)$.
Observe that $\sum_{(e_1, \dots, e_c) \in T} \prod_{j=1}^c \gamma(e_j) \leq \left(\sum_{e_1 \in E(H)} \gamma(e_1)\right) \cdot \cdots \cdot \left(\sum_{e_c \in E(H)} \gamma(e_c)\right)$ as, 
by distributivity, all the terms of the sum on the left-hand side
are also present on the right-hand side of the inequality. Furthermore, we have $\sum_{e \in E(H)} \gamma(e) \leq k$ and thus, by putting it all together, we see that the total weight in $G$ is at most $k^cd$.

From the $B^{\epsilon}(\gamma)$ side, consider an arbitrary vertex $v \in B^{\epsilon}(\gamma)$ and let $e_1, \dots, e_n$ be the
edges in $E(H)$ containing $v$ \revision{with nonzero weight}.  We have $\sum_{j=1}^n\gamma(e_j) \geq \epsilon$ and 
$\gamma(e_j) \leq \frac{\epsilon}{2c}$ for each $j \in [n]$. 
We can apply the above claim for $X = \{\gamma(e_1), \dots, \gamma(e_n)\}$, 
$\delta = \frac{\epsilon}{2c}$, and $w = \epsilon$ to get the inequality
$\sum \gamma(e_{j_1}) \cdot \, \cdots \, \cdot \gamma(e_{j_c})  
\geq (\epsilon -  \frac{\epsilon}{2c} \cdot c)^c = (\frac{\epsilon}{2})^c$, where
the sum ranges over all $c$-tuples 
$(e_{j_1}, \dots, e_{j_c})$
of distinct edges in $E(H)$ containing $v$. 

We conclude that $v$ (now considered as a vertex in $G$) is incident to edges whose total weight is 
$\geq (\frac{\epsilon}{2})^c$ in $E(G)$.
Since we have seen above that the total weight of all edges in $E(G)$  is $\leq k^cd$,
there can be no more than $d(\frac{2k}{\epsilon})^c$ vertices in $B^{\epsilon}(\gamma)$.
\end{proof}

\paragraph{Linear Programs}
We assume some familiarity with Linear Programs (LPs). 
Formally, we are dealing here with minimization problems
of the form $\min \mathbf{c}^T \mathbf{x}$ subject to $\mathbf{A} \mathbf{x} \geq \mathbf{b}$ and $\mathbf{x} \geq \mathbf{0}$, where
$\mathbf{x}$ is a vector of $n$ variables, $\mathbf{c}$ is a vector of $n$ constants, $\mathbf{A}$ is a an $m \times n$ 
matrix, $\mathbf{b}$ is a vector of $m$ constants, and $\mathbf{0}$ stands for the $n$-dimensional zero-vector.
More specifically, for a hypergraph $H = (V,E)$ and vertices $Y \subseteq V$, 
the fractional edge cover number $\rho_H^*(Y)$ of $Y$ is obtained as the optimal value 
of the following LP: let $E = \{e_1, \dots, e_n\}$ \revision{(note that in a slight departure from typical naming, $n$ here is the number of edges in the hypergraph)} and $Y = \{y_1, \dots, y_m\}$, then
$c$ is the $n$-dimensional vector $(1, \dots, 1)$, 
$b$ is the $m$-dimensional vector $(1, \dots, 1)$, and 
$\mathbf{A} \in \{0,1\}^{[m]\times [n]}$, such that $A_{ij} = 1$ if $y_i \in e_j$ and  $A_{ij} = 0$  otherwise.
In the sequel, we will refer to such LPs with 
$\mathbf{c} \in \{1\}^n$, $\mathbf{b}\in \{1\}^m$ and 
$\mathbf{A} \in \{0,1\}^{[m]\times [n]}$
as \emph{unary} linear programs.

For given number $n$ of edges, there are at most $2^n$ possible 
different inequalities of the form $\mathbf{A_i} \mathbf{x} \geq 1$. We thus get the following property of 
unary LPs, which intuitively states that if the optimum is bigger than 
some threshold $k$, then it exceeds $k$ by some distance.

\begin{lemma} \label{dnk}
For all positive integers $n$ and $k$,
there is an integer $D(n,k)$ such that 
for any unary LP $Z$ of at most $n$
variables if $\OPT(Z)>k$ then $\OPT(Z)-k>\frac{1}{D(n,k)}$, 
where $\OPT(Z)$ denotes the minimum of the LP.
\end{lemma}

\section{Bounding the Support of Fractional Edge Covers}
\label{sec:tech}
\subsection{The main combinatorial result}
\label{sec:MainCombinatorialResult}
In this section we establish our main combinatorial result,
Theorem~\ref{maintheor}. Every set of vertices in a $(c,d)$-hypergraph
can be covered in a way such that the size of the 
support depends only on $c$,
$d$, and the set's fractional edge cover number. Recall that we denote by $B(\gamma)$ the set of all vertices $v$ such that $\gamma(v) \geq 1$ where $\gamma$ is a weight function on edges and the weight of a vertex is the sum of all incident edge weights. \revision{For sets $S$ of hyperedges, it will also be convenient to write $\gamma(S)$ for $\sum_{e\in S} \gamma(e)$.}

\begin{theorem} \label{maintheor}
There is a function $h(c,d,k)$ such that the following is true.
Let $H=(V,E)$ be a $(c,d)$-hypergraph and let $\gamma$ be an assignment of weights
to $E$. \revision{Let $k\in \mathbb{Q}^+$ such} that $\weight(\gamma) \leq k$.
Then there exists an assignment $\nu$ of weights to $E$
such that $weight(\nu) \leq k$, $B(\gamma) \subseteq B(\nu)$
and $|\supp(\nu)| \leq h(c,d,k)$.
\end{theorem}

The first step of our reasoning is to consider the situation where $|B(\gamma)|$
is bounded.  In this case it is easy to transform $\gamma$ into the desired
$\nu$. Partition all the hyperedges of $H$ into equivalence classes corresponding
to non-empty subsets of $B(\gamma)$ such that two edges $e_1$ and  $e_2$ are equivalent if and only if $e_1 \cap B(\gamma) = e_2 \cap B(\gamma)$.
Then let $s_X$ be the total weight (under $\gamma$)
of all the edges from the equivalence class where $e \cap B(\gamma) = X$. Identify one representative of each (non-empty) equivalence
class and let $e_X$ be the representative of the equivalence class corresponding to $X$.
Then define $\nu$ as follows. For each $X$ corresponding to a non-empty equivalence class, set
$\nu(e_X)=s_X$. For each edge $e$ whose weight has not been assigned in this way,
set $\nu(e)=0$. It is clear that  $B(\gamma) \subseteq B(\nu)$ and that the support of $\nu$
is at most $2^{|B(\gamma)|}$, which is bounded by assumption.

Of course, in general we cannot assume that $|B(\gamma)|$ is bounded.
Therefore, as the next step of our reasoning, we consider a more general situation
where we have a bounded set ${\bf S}=\{S_1, \dots, S_r\}$ 
where each $S_i$ is a set of at most $c$ hyperedges such that the following conditions hold
regarding ${\bf S}$:\revision{
\begin{enumerate}[label=(\roman*)]
    \item for each $1 \leq i \leq r$, $\gamma(S_i) \geq 1$, and
    \item the set $U=B(\gamma) \setminus \bigcup_{i \in [r]} \bigcap S_i$ is of bounded
size. 
\end{enumerate}
}
Then the assignment $\nu$ as in Theorem \ref{maintheor} can be defined as follows.
For each $e \in \bigcup {\bf S}$, set $\nu(e)=\gamma(e)$. Next, we observe that  for the subhypergraph
 $H'=H - \bigcup {\bf S}$, $|B_{H'}(\gamma)|$ is bounded, where the subscript $H'$ means that we consider $B$ for hypergraph $H'$. Therefore, 
we define $\nu$ on the remaining edges as in the paragraph above.
It is not hard to see that the support of the resulting $\nu$ is of size at most $c \cdot r+2^{|U|}$. 
We are going to show that such a family of sets of edges  can always be found 
for $(c,d)$-hypergraphs (after a possible modification of $\gamma$).

\begin{definition} [{\bf Well-formed pair}]
  Let $H=(V,E)$ be a hypergraph and let $\gamma$ be an edge weight function. We say
 $({\bf S},U)$ is a  \emph{well-formed pair} (with regard to $\gamma$)
if it satisfies the following conditions:
\begin{enumerate}
\item $U \subseteq B(\gamma)$
\item ${\bf S}=\{S_1, \dots, S_r\}$ where each
$S_i$ is a set of at most $c$ hyperedges of $H$.
\item $B(\gamma) \setminus U \subseteq \bigcup_{i \in [r]} \bigcap S_i$. 
\end{enumerate}
We denote $\sum_{i \in [r]} |S_i|+2^{|U|}$ by $n({\bf S},U)$ 
and refer to it as the \emph{size} of $({\bf S},U)$.
\end{definition}

\begin{definition}
A well-formed pair $({\bf S},U)$ is \emph{perfect} if there is an assignment
$\nu : E \to [0,1]$ with $\weight(\nu) \leq k$ and $|\supp(\nu)| \leq n({\bf S },U)$
such that $\bigcup_{i \in [r]} \bigcap { S_i}  \cup U \subseteq B(\nu)$.
  
\end{definition}

Our aim now is to prove the existence of a perfect pair
$({\bf S},U)$ of size bounded by a function depending on $c$, $d$, and $k$.
Clearly,
this will imply Theorem \ref{maintheor}.

In particular, we will define the \emph{initial pair} which is a well-formed pair
but not necessarily perfect. Then we will define two transformations from
one well-formed pair into another and prove existence of a function
$\transf$ so that if $({\bf S_1},U_1)$ is transformed into $({\bf S_2},U_2)$,
then $n({\bf S_2},U_2) \leq \transf(n({\bf S_1},U_1))$.
We will then prove that if we form a sequence of well-formed pairs
starting from the initial pair and obtain every next element by a transformation
of the last one then, after a bounded number of steps we obtain a perfect
well-formed pair. We start by defining the initial pair.

\begin{definition} \label{initpair}
The \emph{initial pair} is $({\bf S_0},U_0)$
where ${\bf S_0}=\{\{e\} \mid \gamma(e) \geq 1/(2c)\}$
and $U_0=B(\gamma) \setminus \bigcup_{\{e\} \in {\bf S_0}} e$. 
\end{definition}

\begin{lemma} \label{initbound} 
There is a function $\init$ such that
$n({\bf S_0},U_0) \leq \init(c,d,k)$
\end{lemma}

\begin{proof}
\revision{We can bound $|U_0|$  by applying Lemma~\ref{smalledges} to all edges with weight less than $1/(2c)$, i.e., to the subhypergraph without the edges in ${\bf S_0}$. In particular,} $|U_0| \leq f(c,d,k)$ where $f$ is as
in Lemma \ref{smalledges} (for $\epsilon=1$)
and $|\bigcup {\bf S_0}| \leq 2ck$ 
by construction.
\end{proof}

We now introduce our two kinds of transformations, \emph{folding} and \emph{extension}. 
A folding removes a set $S^*$ of $c$ edges from 
${\bf S}$ and adds to $U$ the vertices in the intersection
of the edges of $S^*$. In the resulting well-ordered pair
$({\bf S'},U')$, ${\bf S'}$ has one less element than ${\bf S}$
and $U'$, compared to $U$, has a bounded size increase of at most $d$
vertices. Thus the action of folding gets the resulting well-formed
pair closer to one with empty first component, which is a perfect pair
according to the discussion in the beginning of this section. 

\begin{definition} \label{folding}
Let $({\bf S},U)$ be a well-formed pair
such that ${\bf S}$ contains elements of size $c$.
Let $S^* \in {\bf S}$ such that $|S^*|=c$.
Let ${\bf S'}={\bf S} \setminus \{S^*\}$ and
$U'=U \cup (\bigcap S^* \cap B(\gamma))$.
We call $({\bf S'},U')$ a \emph{folding} of $({\bf S},U)$. 
\end{definition}

The folding, however, is possible only if ${\bf S}$ has an element of
size $c$. Otherwise, we need a more complicated transformation called
an \emph{extension}.  The extension takes an element $S \in {\bf S}$
of size $\revision{c'}<c$ and expands it by replacing $S$ with several subsets of
$E(H)$ each containing all the edges of $S$ plus one extra edge.  This
replacement may miss some of the elements $v$ of
$B(\gamma) \cap \bigcap S$ simply because $v$ is not contained in any
of these extra edges.  This excess of missed elements is added to
$U$ and thus all the conditions of a well-formed pair are satisfied.

\begin{definition} \label{defext}
  Let $({\bf S},U)$ be a well-formed pair
with ${\bf S}  \neq \emptyset$ such that every element of ${\bf S}$ is of size at most $c-1$.
For the extension, let $S \in {\bf S}$ be an element called the \emph{extended element}
and let a set $S'$ of hyperedges be called the \emph{extending set}.
We refer to $L=(\bigcap S \cap B(\gamma) ) \setminus \bigcup S'$ as the set of \emph{light vertices}.
An \emph{extension}  of $({\bf S},U)$ is $({\bf S'},U')$  where
${\bf S'}=({\bf S} \setminus \{ S\})  \cup \{ S \cup \{e\}| e \in S'\}$
and $U'=U \cup L$.   

\end{definition}

\begin{proposition} \label{validext}
With data as in Definition \ref{defext},
$({\bf S'},U')$ is a well-formed pair.
\end{proposition}

At the first glance the transformation performed by the extension is radically
opposite to the one done by the folding: the first component grows rather than
shrinks. Note, however, that the new sets replacing the removed one 
contain a larger number of edges and thus they are closer to being of size $c$
at which stage the folding can be applied to them. The intuition is that after a sufficiently
large number of foldings and extensions, a well-formed pair with empty first component
is eventually obtained.

For our overall goal, we then need to show that the size of the resulting
perfect pair is indeed bounded by a function of $c$, $d$, and $k$.
To that end, the following lemma first establishes that a single step in this process
increases the size of the well-formed pair in a controlled manner.
To streamline our path to the main result, the proof of the 
lemma is deferred to Section~\ref{sec:boundext}.

\begin{restatable}{lemma}{lemboundext}\label{lemboundext}
There is a function $ext$ such that the following holds.
Let $({\bf S},U)$ be a well-formed pair
with ${\bf S}  \neq \emptyset$ such that every of element of ${\bf S}$ is of size at most $c-1$.
Then one of the following two statements is true.
\begin{enumerate}
\item $({\bf S},U)$ is a perfect pair.
\item There is an extension $({\bf S'},U')$ of $({\bf S},U)$ such that
$n({\bf S'},U') \leq ext(n({\bf S},U))$.
We refer to $({\bf S'},U')$ as a \emph{bounded extension} of $({\bf S},U)$
\end{enumerate}
\end{restatable}

For the sake of syntactical convenience, we unify the notions of folding and
bounded extension into a single notion of transformation and prove the related
statement following from Lemma \ref{lemboundext} and the definition of folding.

\begin{definition}
Let $({\bf S},U)$ and $({\bf S'},U')$ be well-formed pairs.
We say that $({\bf S'},U')$ is a \emph{transformation} of $({\bf S},U)$
if it is either a folding or a bounded extension of $({\bf S},U)$.
\end{definition}

\begin{lemma} \label{boundtransf}
There is a monotone function $\transf$ with $\transf(x) \geq x$ for any 
natural number $x$ such that the following holds. 
If $({\bf S},U)$ is a well-formed pair, then one of the following
two statements is true.
\begin{enumerate}
\item $({\bf S},U)$ is a perfect pair.
\item There exists a transformation $({\bf S'},U')$ of $({\bf S},U)$
such that \\$n({\bf S'},U') \leq \transf(n({\bf S},U))$.
\end{enumerate}
\end{lemma}

\begin{proof}
Assume that $({\bf S},U)$ is not a perfect pair.
Then $|{\bf S}|$ is not empty (see the discussion at the beginning of this section).
Suppose that an element of ${\bf S}$ is of size $c$.
Then we set $({\bf S'},U')$ to be a folding of $({\bf S},U)$.
By definition of the folding and of $(c,d)$-hypergraphs,
$({\bf S'},U')$ is obtained from $({\bf S},U)$ by removal of an element
from ${\bf S}$ and adding at most $d$ vertices to $U$.
Hence the size of $({\bf S'},U')$ is clearly bounded in 
the size of $({\bf S},U)$.
If all elements of ${\bf S}$ are of size at most $c-1$
then by Lemma \ref{lemboundext}, 
there is a bounded extension $({\bf S'},U')$ of $({\bf S},U)$.

Clearly, we can specify a function $\transf'$ so that in both
cases $n({\bf S},U) \leq \transf'(n({\bf S},U))$.
In particular, to satisfy the requirement for $\transf$, it suffices to
set $\transf(x)=\max(x, \max_{i \in [x]} \transf'(\revision{i}))$
for each natural number $x$.
\end{proof}

Now that we know that each individual step on our path to a perfect pair increases the size only in
a bounded fashion, we need to establish that the number of steps is also bounded by a function of $c$, $d$, and $k$.
The following auxiliary theorem states that such a bound exists. 
\begin{definition}
A sequence of $({\bf S_1},U_1), \dots, ({\bf S_q}, U_q)$ 
is a \emph{sequence of transformations} if 
for each $i \in [q-1]$ the following two statements hold
\begin{enumerate}
\item $({\bf S_i},U_i)$ is not a perfect pair.
\item $({\bf S_{i+1}},U_{i+1})$  is a transformation of $({\bf S_i},U_i)$ \revision{as in Lemma~\ref{boundtransf}}.
\end{enumerate}
\end{definition}

\begin{restatable}{theorem}{seqsizethm} \label{seqsize}
There is a monotone function $sl$ such that the following is true.
Let $({\bf S_1},U_1), \dots, ({\bf S_q}, U_q)$  be a sequence 
of transformations.
Then $$q \leq sl(n({\bf S_1},U_1)).$$
\end{restatable}

The proof of Theorem \ref{seqsize} is provided in Section \ref{sec:length}. 

In summary, we have shown that we can reach a perfect pair in a bounded number of transformations. Moreover, each transformation
increases the size of a pair in a controlled manner. We are now ready to prove our main result.

\begin{proof}[Proof of Theorem~\ref{maintheor}]
Consider the following algorithm.
\begin{enumerate}
\item Let $({\bf S_0},U_0)$ be the initial pair (see Definition \ref{initpair}).
\item $q \leftarrow 0$
\item While $({\bf S_q},U_q)$ is not a perfect pair
       \begin{enumerate}
       \item $q \leftarrow q+1$
       \item Let $({\bf S_q},U_q)$ be a transformation of $({\bf S_{q-1}},U_{q-1})$, which exists by Lemma~\ref{boundtransf} 
       \end{enumerate}
\end{enumerate}
By Theorem \ref{seqsize}, the above algorithm stops with the final
$q$ being no higher than $sl(n({\bf S_1},U_1))$.
It follows from the description of the algorithm that $({\bf S_q},U_q)$ is a 
perfect pair. It remains to show that its size is bounded by a function
of $c,d,k$.
\begin{equation} \label{eqq1}
    q \leq  sl(n({\bf S_0},U_0)) \leq sl(\init(c,d,k))
\end{equation}
the second inequality follows from 
Lemma \ref{initbound} and the monotonicty of 
 $sl$.
Next, by the properties  of $\transf$, an inductive application
of Lemma~\ref{boundtransf} and Lemma~\ref{initbound} yields
\begin{equation}  \label{eqq2}
n({\bf S_q},U_q) \leq \transf^{q}(\init(c,d,k)) 
\end{equation}
where superscript $q$ means that function $\transf$ is composed with itself $q$-times, that is $\transf(\transf(\transf(...)))$.

Let $h(c,d,k)=\transf^{sl(\init(c,d,k))}(\init(c,d,k))$.
It follows from combination of \eqref{eqq1} and \eqref{eqq2}
that $n({\bf S_q},U_q) \leq h(c,d,k)$.
\end{proof}

 \subsection{Proof of Lemma \ref{lemboundext}} \label{sec:boundext}

The first step of the proof is to define a unary linear program of bounded
size associated with $({\bf S},U)$.
Then we will demonstrate that if the optimal value of this linear program
is at most $k$, then $({\bf S},U)$ is perfect. Otherwise,
we show that a bounded extension can be constructed. 
 
In order to define the linear program, we first formally define 
equivalence classes of edges covering $U$ (see the informal discussion
in Section~\ref{sec:MainCombinatorialResult}).  

\begin{definition}[\bf Working subset]
A set of vertices $U' \subseteq U$ is called \ \emph{working subset} 
(for $({\bf S},U)$) if there is $e \in E(H) \setminus \bigcup {\bf S}$
such that $e \cap U=U'$. This $e$ is called a \emph{witnessing edge}
of $U'$ and the set of all witnessing edges of $U'$ is denoted by $W_{U'}$.
\end{definition}

Continuing on the previous definition,
it is not hard to see that the sets $W_{U'}$ partition the set of 
edges of $E(H) \setminus \bigcup {\bf S}$ having a non-empty 
intersection with $U$. Choose an arbitrary but fixed representative
of each $W_{U'}$ and let $A_U$ be the set of these representatives
which we also refer to as the set of \emph{witnessing representatives}.
Now, we are ready to define the linear program.

\begin{definition}[$LP({\bf S},U)$]
\label{def:lpSU}
The linear program $LP({\bf S},U)$ of $({\bf S},U)$ has 
the set of variables $X=\{x_e \mid e \in \bigcup {\bf S} \cup A_U\}$. 
The objective function is the minimization
of $\sum_{x_e \in X} x_e$. The constraints are of the following
three kinds.
\begin{enumerate}
\item $\{ 0 \leq x_e \leq 1 \mid x_e \in X \}$.
\item $\{One_S \mid S \in {\bf S}\}$ where $One_S$ is $\sum_{e \in S} x_e \geq 1$.
\item $\{One_u \mid u \in U\}$ where $One_u$ is $\sum_{e \in E_u} x_e \geq 1$
where $E_u$ in turn is the subset of $\bigcup {\bf S} \cup A_U$ consisting of all the edges
containing $u$. 
\end{enumerate}
\end{definition}

\begin{lemma} \label{properf}
Assume that the optimal solution of $LP(S,U)$ is at most $k$.
Then $({\bf S},U)$ is a perfect pair.
\end{lemma}

\begin{proof}
Each variable $x_e$ of $LP({\bf S},U)$  corresponds to an edge $e$
and this correspondence is injective.
For each $x_e$, let $\nu(e)$ be the value of $x_e$ in the optimal solution.
For each edge $e$ not having a corresponding variable, set $\nu(e)=0$. Note that $\nu$ is an edge weight function with total weight at most $k$.
It follows from a direct inspection that
$U \cup \bigcup_{i \in [r]}  \bigcap {\bf S_i}  \subseteq B(\gamma)$
and the size of the support of $\nu$ is at most $n({\bf S},U)$.
\end{proof}

\revision{The proof of Lemma~\ref{properf} establishes a correspondence between solutions of $LP(S,U)$ and edge weight functions. We will implicitly extend notions for edge weight functions (like their weight) to solutions of $LP(S,U)$ via this correspondence for the rest of this section.}

As stated above, in case the optimal value of $LP({\bf S},U)$  is greater than
$k$, we are going to demonstrate the existence of a bounded extension of $({\bf S},U)$.
The first step towards identifying such an extension is to identify the 
extended element of ${\bf S}$. Combining Lemma~\ref{dnk} from Section~\ref{sec:prelim} 
with Lemma~\ref{upper} below, we observe that ${\bf S}$ has an element 
$S^*$ such that $\gamma(S^*)$ is \revision{bounded away from} 
$1$.
This $S^*$ will be the extended element. 

\begin{lemma} \label{upper}
 Let $({\bf S},U)$ be a well-formed pair.
Let ${\bf S^*}$ be the subset of ${\bf S}$ consisting 
of all $S$ such that $\gamma(S)<1$.
Let $\alpha$ be an optimal solution for $LP({\bf S},U)$.
Then $\weight(\alpha) \leq \weight(\gamma)+\sum_{S \in {\bf S^*}} (1-\gamma(S))$. 
\end{lemma}

\begin{proof}
Let $\beta$ be an arbitrary assignment of weights to the hyperedges
of $H$. We say that $\beta$ satisfies a constraint $One_S$ 
for $S \in {\bf S}$ if $\beta(S) \geq 1$
and that $\beta$ satisfies the constraint $One_u$ for
$u \in U$ if $\beta(E_u) \geq 1$.

We are going to demonstrate an assignment of weights whose
total weight exceeds that of $\gamma$ by at most  
$\sum_{S \in {\bf S^*}} (1-\gamma(S))$ and that satisfies all the
constraints $One_S$ and $One_v$. Clearly, this will imply
correctness of this theorem.

For each $S \in {\bf S^*}$, choose an arbitrary edge $e_S \in S$
and let $\INCR$ be the set of all such edges. 
For each $e \in \INCR$, let $incr_e=\max\{1-\gamma(S) \mid e=e_S\}$. \revision{That is, $e \in \INCR$ can be the representative of multiple $S \in {\bf S^*}$ and $incr_e$ represents the maximal $1-\gamma(S)$ over all the sets $S$ for which $e$ is the representative $e_S$.}

Let $\gamma'$ be obtained from $\gamma$ as follows.
If $e \in \INCR$ then $\gamma'(e)=\gamma(e)+incr_e$.
Otherwise, $\gamma'(e)=\gamma(e)$.
It is not hard to see that $\gamma'$ satisfies the constraints
$One_S$ for each $S \in {\bf S}$, and that 
$\weight(\gamma') \leq \weight(\gamma)+\sum_{S \in {\bf S^*}} (1-\gamma(S))$. \revision{Since $\gamma'$ does not decrease the weight of any edge, we also observe
$U \subseteq B(\gamma')$.}

Let $\{U_1, \dots, U_a\}$ be all the working subsets of $U$
and let $e_1, \dots, e_a$ be the respective witnessing representatives. \revision{The set of edges $e_1, \dots, e_a$ corresponds to the set $A_U$ from Definition~\ref{def:lpSU}}.
Then the assignment $\gamma''$ of weights is defined as follows.
\begin{enumerate}
\item If there is an $i \in [a]$ such that $e \in W_{U_i}$ \revision{and $e=e_i$}, then
$\gamma''(e)=\gamma'(W_{U_i})=\gamma(W_{U_i})$.%
\item \revision{If $e \in \bigcup \mathbf{S}$, then $\gamma''(e)=\gamma'(e)$.}
\item \revision{Otherwise, $\gamma''(e)=0$}.
\end{enumerate}

Let $W=\bigcup_{i \in [a]} W_{U_i}$.
Note that, by construction, $\gamma'(W)=\gamma''(W)$
and the weights of edges outside $W$ are the same 
under $\gamma'$ and $\gamma''$ and thus, $\weight(\gamma')=\weight(\gamma'')$.
Moreover since $\bigcup {\bf S}$ does not intersect with $W$,
$\gamma''$ satisfies the constraints $One_S$ for all $S \in {\bf S}$.

It remains to show that $\gamma''$ satisfies the constraints $One_u$
for each $u \in U$.
Let $e_1, \dots, e_r$ be the edges of $\bigcup {\bf S}$ containing $u$, let
$\{U_1, \dots, U_b\}$ be the working subsets of $U$ containing $u$,
and let $e'_1, \dots, e'_b$ be the respective witnessing representatives.
As $u \in B(\gamma')$, it follows
that $\sum_{i \in [r]} \gamma'(e_i)+\sum_{i \in [b]} \gamma'(W_{U_i}) \geq 1$.
By construction, $\gamma''(e_i)=\gamma'(e_i)$ for each $1 \leq i \leq r$
and $\gamma''(e'_i)=\gamma'(W_{U_i})$ for each $1 \leq i \leq b$.
Consequently,
 $\sum_{i \in [r]} \gamma''(e_i)+\sum_{i \in [b]} \gamma''(e'_i) \geq 1$.
We conclude that $\gamma''$ satisfies $One_u$.
\end{proof}

\revision{Recall that Lemma~\ref{dnk} states that for integers $n$ and $k$, and unary LP $Z$ with at most $n$ variables and $OPT(Z) > k$, there is an integer $D(n,k)$ such that $OPT(Z) - k > \frac{1}{D(n,k)}$.}
Together with Lemma~\ref{upper} this implies the following corollary.

\begin{corollary} \label{corwitness}
Let $({\bf S},U)$ be a well-formed pair.
Assume that $weight(\gamma) \leq k$ while
$OPT(LP({\bf S},U))>k$.
Let $n=n({\bf S},U)$.  
Then there is an $S^* \in {\bf S}$ with
$1-\gamma(S^*)>\frac{1}{D(n,k)\cdot |{\bf S}|}$. 
In particular this means that ${\bf S^*}$ is not empty
where ${\bf S^*}$ is as in Lemma \ref{upper}.
\end{corollary}

\begin{proof}
Note that the number of variables of $LP({\bf S},U)$ is at most $n$.
It follows from the combination of Lemma \ref{dnk} and Lemma \ref{upper}
that $\weight(\gamma)+\sum_{S \in {\bf S^*}} (1-\gamma(S))>k+1/D(n,k)$ and,
since $\weight(\gamma) \leq k$,
$\sum_{S \in {\bf S^*}} (1-\gamma(S))>1/D(n,k)$ and hence there is $S^* \in {\bf S^*}$
with $(1-\gamma(S^*)) > \frac{1}{D(n,k) \cdot|{\bf S^*}|} \geq \frac{1}{D(n,k) \cdot |{\bf S}|}$.
\end{proof}

\revision{
We are now ready to prove Lemma~\ref{lemboundext}. Let us first recall the lemma and note that the first component of the well-formed pair in the statement is non-empty and its elements contain at most $c-1$ edges.
\lemboundext*
}
\begin{proof}[Proof of Lemma~\ref{lemboundext}]
If the value of the optimal solution of $LP({\bf S},U)$ is at most $k$,
we are done by Lemma~\ref{properf}.

Otherwise, let $S^* \in {\bf S}$ be as in Corollary \ref{corwitness}.
Let $\epsilon=(D(n,k)\cdot |{\bf S}|)^{-1}$.
It follows from Corollary~\ref{corwitness} that 
vertices of $B(\gamma) \cap \bigcap S^*$ need weight contribution of at least
$\epsilon$ from hyperedges of $H$ other than $S^*$.
We define the extending set $S'$ as the set of all hyperedges of 
$H$ other than $S^*$ whose weight is at least $\epsilon/2c$ and therefore $|S'|\leq 2ck / \epsilon$.
\revision{We observe that the set $L$ of light vertices (cf., Definition~\ref{defext}) is the }subset
of $B(\gamma) \cap \bigcap S^*$ consisting of all vertices $x$ that, besides
$S^*$ are contained only in hyperedges of weight smaller than $\epsilon/2c$.
By Lemma~\ref{smalledges}, $|L| \leq f(c,d,k)$ and the size of $S^*$ is \revision{at most $c-1$ by assumption.}
It is not hard to see that the size of the resulting extension is bounded
as well.
\end{proof}

 \subsection{Proof of Theorem \ref{seqsize}} \label{sec:length}

For this theorem, rather than considering 
a well-formed pair $({\bf S},U)$ itself we consider
the pair $(A,b)$ where $A$ is the multiset of sizes of the sets
of ${\bf S}$ and $|U|=b$. We call $(A,b)$ a \emph{bare bones $c$-pair} ($c$-BBP). 
A transformation of $({\bf S},U)$ is translated into a bounded size transformation
of $({A},b)$.  In the next five definitions we formalize this intuition.
Then we state Theorem \ref{barebound} claiming that  a sufficiently
long sequence of bounded transformations of $c$-BBPs results in one
where the first component is empty. This will imply Theorem \ref{seqsize}
because a $c$-BBP with the empty first components is translated
back into well-formed pair with the empty first component which is 
perfect.  Finally, we prove Theorem \ref{barebound}.

\begin{definition} \label{barebones}
A \emph{bare bones $c$-pair}, abbreviated as $c$-BBP
is a pair $(A,b)$ where $A$ is a multiset of integers in the range $[1,c]$
and $b$ is just a non-negative integer.
We denote $2^b+\sum_{x \in A} x$ by $n(A,b)$.
Note that the number of occurrences of each
$x \in A$ in the sum is its multiplicity in $A$.
\end{definition}

\begin{definition}
Let $(A,b)$ be a $c$-BBP and assume that $c \in A$.
Let $A'=A \setminus \{c\}$ (that is, the multiplicity of $c$
in $A$ is reduced by one) and let $b'=b+d$
where $d$ is a non-negative integer. 
Clearly $(A',b')$ is a $c$-BBP, we refer to it as a \emph{folding}
of $(A,b)$.
\end{definition}

\begin{definition}
Let $(A,b)$ be a $c$-BBP and let $x \in A$ such that $x<c$.
Let $A'$ be obtained from $A$ by removal of one occurrence of $x$ 
and adding $d_1$ occurrences of $x+1$ for some non-negative integer $d_1$.
Let $b'=b+d_2$ for some non-negative integer $d_2$.
Clearly $(A',b')$ is a $c$-BBP, we refer to it as an \emph{extension}
of $(A,b)$
\end{definition}

\begin{definition}
Let $(A,b)$ and $(A',b')$ be $c$-BBPs such that $(A',b')$
is either a folding or an extension of $(A.b)$. We then say that
$(A',b')$ is a \emph{transformation} of $(A,b)$.
Let $n=n(A,b)$ and $n'=n(A',b')$ and suppose that
$n' \leq g(n)$ for some function $g$. We then say that
$(A',b')$ is a $g$-\emph{transformation} of $(A,b)$.
\end{definition}

\begin{definition}
Let $g$ be a function of one argument
and let $(A_1,b_1), \dots, (A_r,b_r)$ be a sequence of $c$-BBPs
such that for each $2 \leq i \leq r$, the $c$-BPP
$(A_i,b_i)$ is a $g$\nobreakdash-transformation of $(A_{i-1},b_{i-1})$.
We call $(A_1,b_1), \dots, (A_r,b_r)$ a $g$-transformation sequence.
Note that for each $1 \leq i <r$, $A_i$ is not empty for otherwise,
it is impossible to apply a transformation to $(A_i,b_i)$.
\end{definition}

\begin{theorem} \label{barebound}
Let $g$ be a function of one argument.
Then there is a function $h[g]$ such
that if $(A_1,b_1), \dots, (A_r,b_r)$ is a $g$-transformation
sequence then $r \leq h[g](n)$ where $n=n(A_1,b_1)$.
\end{theorem}

We first show how to prove Theorem~\ref{seqsize} using
Theorem~\ref{barebound} and then we will prove Theorem~\ref{barebound}
itself. \revision{We first recall the theorem.
\seqsizethm*
}
\begin{proof}[Proof of Theorem~\ref{seqsize}]
Let $({\bf S},U)$ be a well-formed pair and let $bbp({\bf S},U)$ be $(A,b)$ where
$A$ is the multiset of sizes of elements of ${\bf S}$
(each $x$ occurs in $A$ exactly the number of times as there are
sets of size $x$ in ${\bf S}$) and $b=|U|$.
It is not hard to see that $(A,b)$ is a $c$-BBP. 
Moreover,
\begin{equation} \label{eqbare}
n({\bf S},U)=n(A,b)
\end{equation} 
Let $({\bf S_1},U_1), \dots, ({\bf S_r},U_r)$ be a 
transformation sequence.
Let $(A_1,b_1), \dots, (A_r,b_r)$ be a sequence of $c$-BBPs such that
$(A_i,b_i)=bbp({\bf S_i},U_i)$ for each $1 \leq i \leq r$.

We are going to show that 
$(A_1,b_1), \dots, (A_r,b_r)$ is a $\transf$-transformation sequence.
By Theorem \ref{barebound}, this will imply that
$r \leq h[\transf](n)$
where $n=n(A_1,b_1)=n({\bf S_1},U_1)$ by \eqref{eqbare}
thus implying the theorem. 

So, consider two arbitrary consecutive elements $(A_i,b_i)$ and
$(A_{i+1},b_{i+1})$.

Assume first that $({\bf S_{i+1}},U_{i+1})$ is obtained from
$({\bf S_i},U_i)$ by folding. It is not hard to see that
$(A_{i+1},b_{i+1})$ is obtained from $(A_i,b_i)$ by removing
one occurrence of $c$ and adding $b_{i+1}=b_{i}+(|U_{i+1}|-|U_i|)$.
That is $(A_{i+1},b_{i+1})$ is obtained from $(A_i,b_i)$ as result
of folding. 
As $n({\bf S_{i+1}}, U_{i+1}) \leq \transf(n({\bf S_i},U_i))$, 
it follows from \eqref{eqbare} that 
$n({A_{i+1}}, b_{i+1}) \leq \transf(n({A_i},b_i))$.
We conclude that $(A_{i+1},b_{i+1})$ is obtained from $(A_i,b_i)$
as a result of a $\transf$-transformation.

Assume now that $({\bf S_{i+1}},U_{i+1})$ is obtained from
$({\bf S_i},U_i)$ by extension. This means that 
${\bf S_{i+1}}$ is obtained from ${\bf S_i}$ by removal of
some $S^*$ of size less than $c$ and replacing it with
$d_1$ sets of size $c+1$ for some integer $d_1 \geq 0$.
Also $U_{i+1}$ is obtained from $U_i$ by adding $d_2$
new elements for some integer $d_2 \geq 0$. 
It follows by construction that $(A_{i+1},b_{i+1})$
is an extension of $(A_i,b_i)$. By the same argumentation as
in the end of the previous paragraph, we conclude that
$(A_{i+1},b_{i+1})$ is obtained from $(A_i,b_i)$ by 
$\transf$-transformation.
\end{proof}

\begin{proof}[Proof of Theorem~\ref{barebound}]
We assume w.l.o.g. that $g$ is monotone that is
for $n_1<n_2$ $g(n_1) \leq g(n_2)$.
Indeed, otherwise, since $g$ is defined over non-negative integer,
we can define $g^*(n)$ as the maximum over $g(0), \dots, g(n)$ and
use $g^*$ instead of $g$. The monotonicity allows us to
derive the following inequality.

Suppose that
$(A_1,b_1), \dots, (A_x,b_x)$ is a $g$-transformation sequence
and $x \leq y$. Then
\begin{equation} \label{grow1}
n(A_x,b_x) \leq g^{(y)}(n(A_1,b_1))
\end{equation}
For $i \in \{0, \dots, c-1\}$, a $g$-transformation
is subset of $q$-transformations with an additional
property recursively defined as follows.

\begin{enumerate}
\item $(A',b')$ is a $(g,0)$-transformation of $(A,b)$
if $(A',b')$ is obtained from $(A,b)$ by folding. 
\item Suppose $i>0$ and $(g,i-1)$-transformation has been defined
Then $(A',b')$ is a $(g,i)$-transformation of $(A,b)$
if it is either a $(g,i-1)$-transformation or an extension
where the element removed from $A$ is \revision{an occurrence of} $c-i$. 
\end{enumerate}

A $(g,i)$-transformation sequence $(A_1,b_1), \dots, (A_r,b_r)$
where for each $2 \leq j \leq r$, $(A_j,b_j)$ is obtained from
$(A_{j-1},b_{j-1})$ by $(g,i)$-transformation.
The sequence is \emph{final} if all elements of $A_r$ are smaller than $c-i$,
that is no further extension of the sequence is possible.

We prove by induction that for each $i \in \{0, \dots, c-1\}$,
there is a function $h_i[g]$ such that $r$ as above is at most
$h_i[g](n(A_1,b_1))$. Then $h_{c-1}[g]$ will be the desired function
$h[g]$. For the sake of simplicity, we will omit $g$ in the square brackets
and refer to these functions as $h_0, \dots, h_{c-1}$.

\revision{It is important to observe that our induction is from above, in the sense that we with an increase in $i$ we allow for the removal of occurrences of lower values. As a consequence, a $(g,i)$-transformation of $(A,b)$ can never increase the number of occurrences $c-i$ in $A$: a folding only removes values from $A$ and an extension can only introduce $c-i$ by removing occurrences of $c-i-1$, which is not permitted in a $(g,i)$-transformation. It follows that in any $(g,i)$-transformation sequence starting at $(A_1,b_1)$, there can be at most $m$ extensions that remove an occurrence of $c-i$, where $m$ is the number of $c-i$'s in $A_1$. This observation is key to the following argument.}

The existence of function $h_0$ is easy to observe.
Indeed, the number of consecutive foldings is at most
the multiplicity of $c$ in $A_1$. So, we can put $h_0=n(A_1,b_1)$.

Assume now that $i>0$ and that 
$(A_1,b_1), \dots, (A_r,b_r)$ is a $(g,i)$-transformation sequence.
If it is in fact a $(g,i-1)$-transformation sequence then 
$r \leq h_{i-1}(A_1,b_1)$ by the induction assumption.
Otherwise, let $1<x_1< \dots < x_a \leq r$ be all the indices 
such that for each $1 \leq j \leq a$, $(A_{x_j},b_{x_j})$ is obtained
from $(A_{x_j-1},b_{x_j-1})$ by extension removing an element $c-i$.

For the sake of succinctness, denote $n(A_1,b_1)$ by $n$
and for each $1 \leq j \leq a$, we denote $n(A_{x_j},b_{x_j})$ by $n_j$.
 
For each integer $j \geq 1$, define function $f_j$ as follows.
$f_1(x)=h_{i-1}(x)+1$.
Suppose that $j>1$ and that $f_{j-1}$ has been defined.
Then $f_j(x)=f_{j-1}(x)+h_{i-1}(g^{(f_{j-1}(n))}(x))$.

We show that for each $1 \leq j \leq a$,
$x_j \leq f_j(n)$. 
Note that $(A_1,b_1), \dots$, $(A_{x_1-1}, b_{x_1-1})$ is
a $(g,i-1)$-transformation sequence. Hence, by the induction
assumption, $x_1-1 \leq h_{i-1}(n)$ and $x_1 \leq f_1(n)$.

Furthermore, let $j>1$.
Then $(A_{x_{j-1}},b_{x_{j-1}}), \dots, (A_{x_j-1},b_{x_j-1})$
is also a $(q,i-1)$-transformation sequence.
Therefore, by the induction assumption,
$x_j \leq x_{j-1}+h_{i-1}(n_{j-1})$.
By the induction assumption, $x_{j-1} \leq f_{j-1}(n)$ and, 
by \eqref{grow1}, $n_{j-1} \leq g^{(f_{j-1}(n))}(n)$.
Therefore, 
$x_j \leq f_{j-1}(n)+h_{j-1}(g^{(f_{j-1}(n))}(n))=f_j(n)$ as required.
Applying the same argumentation to the sequence following
$(A_{x_a},b_{x_a})$, we conclude that $r \leq f_{a+1}$.
\revision{As noted above, $(g,i)$-transformations cannot introduce new occurrences of $c-i$ and thus $a$ is at most the number of occurrences of $c-i$ in $(A_1,b_1)$. We can generously bound $a$ by $n$ and conclude that $r \leq f_{n+1}$.}
Hence, we can set $h_i=f_{n+1}$.
\end{proof}

\section{Applications and Extensions}
\label{sec:apps}

\subsection{Checking Fractional Hypertree Width}
\label{sec:fhw}

Now that our main combinatorial result has been established we move
our attention to an algorithmic application of the support bound. In
particular, we are interested in the problem of deciding whether for an
input hypergraph $H$ and constant $k$ we have $\fhw(H) \leq k$. 
The
problem is known to be \np-hard even for $k=2$~\cite{DBLP:conf/pods/FischlGP18}. However, as
noted in the introduction, it has recently been shown that for hypergraph classes which enjoy bounded intersection
or bounded degree, it is indeed tractable to check $\fhw(H) \leq k$ for
constant $k$~\cite{JACM2021}.

\revision{Here we show that our main combinatorial result reveals a large class of instances, that subsumes and extends all previously known cases, for which checking $fhw$ is tractable.} To establish the result we make use of the framework for tractable width checking developed in~\cite{JACM2021}.
We will only recall the
necessary key components here 
and use them in a black-box fashion. 

\begin{definition}
  Let $\rho_q^*(U)$ be the minimal weight of an assignment $\gamma$  such that $U \subseteq B(\gamma)$ and $|\supp(\gamma)| \leq q$.
  We define the $q$-\emph{limited} fractional hypertree width of a hypergraph $H$ as its $\rho^*_q$-width.
\end{definition}

\begin{lemma}[Theorem 4.5 \& Lemma 6.5 in \cite{JACM2021}]
  \label{boundedfhw}
  Fix $c$, $d$, and $q$ as constant integers. There is a polynomial-time algorithm testing whether a given
  $(c,d)$-hypergraph has $q$-limited fractional hypertree width at
  most $k$.
\end{lemma}

The underlying intuition of $q$-limited $\fhw$ is that the bounded
support allows for a polynomial-time enumeration of all the
(inclusion) maximal covers of sufficient weight. For $(c,d)$-hypergraphs, it
is then possible to compute a set of candidate bags such that a
fitting tree decomposition, if one exists, uses bags only from this
set. Deciding whether a tree decomposition can be created from a given
set of candidate bags is tractable under some minor restrictions to
the structure of the resulting decomposition (not of any concern to
the case discussed here).

\revision{
Recall, a class ${\cal C}$ of hypergraphs is said to satisfy the {\em bounded multi-intersection property (BMIP)} if there exist
$c \geq 2$ and $d \geq 0$, such that every $H$ in ${\cal C}$ is a $(c,d)$-hypergraph.}
We now apply our main result and show that, under BMIP, there exists a
constant $q$ such that the $q$-limited fractional hypertree width
always equals fractional hypertree width. From the previous lemma
it is then straightforward to arrive at the desired tractability result.

\begin{theorem} \label{mainres}
There is a polynomial-time algorithm for testing whether the $\fhw$ of the given
$(c,d)$-hypergraph $H$ is at most $k$ (the degree of the polynomial is upper bounded
by a fixed function depending on $c,d,k$).
\end{theorem}
\begin{proof}
  It follows from Theorem~\ref{maintheor}
that if $\fhw(H) \leq k$ for a $(c,d)$-hypergraph $H$
then the $h(c,d,k)$-limited $\fhw$ of $H$ is also
at most $k$.

Indeed, let $\tdecomp$ be a tree decomposition
with $\fhw$ at most $k$. Then, according to Theorem~\ref{maintheor},
for each node $u$ in $T$ there is an edge weight function $\gamma$  with $|\supp(\gamma)|\leq h(c,d,k)$ such that $B_u \subseteq B(\gamma)$.
In other words, it follows that $\tdecomp$
has $\rho^*_q$-width at most $k$ where $q$ is $h(c,d,k)$.
Thus, $H$ also has $h(c,d,k)$-limited
fractional hypertree width at most $k$. \revision{For completeness of the procedure, note that the $h(c,d,k)$-limited fractional hypertree width can never be lower than $\fhw(H)$.}

To test whether $\fhw(H) \leq k$
it is therefore enough to test whether the $h(c,d,k)$-limited $\fhw$ of $H$ is at most $k$. 
This can be done in a polynomial time according to Lemma~\ref{boundedfhw}.
\end{proof}

\subsection{Extension to Fractional Hitting Set}
\label{sec:duals}
\label{sec:vc}

There are two natural dual concepts of fractional edge covers. One is the notion of \emph{fractional hitting sets} which is dual in the sense that it is equivalent to the fractional edge cover on the dual hypergraph. The other, \emph{fractional independent sets}, corresponds to the dual linear program of a linear programming formulation of fractional edge covers. Here we discuss how our results extend to hitting sets. %

We start by giving a formal definition of the fractional hitting set problem. Let $H = (V,E)$ be a hypergraph and $\beta : V \to [0,1]$ be an assignment of weights to the vertices of $H$. Analogous to the definition of fractional edge covers we define
\begin{itemize}
\item $B_v(\beta) = \{e \in E \mid \sum_{v \in e} \beta(v) \geq 1\}$,
\item $\vsupp(\beta) = \{v \in V \mid \beta(v) > 0\}$,
\item and $\weight(\beta) = \sum_{v \in V} \beta(v)$.
\end{itemize}

A fractional hitting set is also called a fractional transversal in some
contexts (cf.~\cite{fractionalGraphTheory}). For a set of edges $E'$,
we denote the weight of the minimal fractional hitting set $\beta$
such that $E' \subseteq B_v(\beta)$ as $\tau^*(E')$. For
hypergraph $H=(V,E)$, we say $\tau^*(H) = \tau^*(E)$. Recall, that
we assume reduced hypergraphs and therefore there is a one-to-one
correspondence of vertices in $H$ and edges in $H^d$. We will make use
of the following straightforward observations about the connection of what we will
call \emph{dual weight assignments}.
\begin{proposition}
  \label{prop:dualcover}
  Let $H=(V,E)$ be a (reduced) hypergraph and let $H^d=(W,F)$ be its dual. We write $f_v$ to identify the edge in $F$ that corresponds to the vertex $v$ in $V$. The following two statements hold:
  \begin{itemize}
  \item   For every $\gamma: E \to [0,1]$ and the function $\beta : W \to [0,1]$ with $\beta(e)=\gamma(e)$ it holds that
  $B_v(\beta) = \{f_v \mid v \in B(\gamma)\}$\footnote{Recall that the edges $E$ of $H$ are the vertices $W$ of $H^d$.}.
\item For every $\beta : V \to [0,1]$ and the function $\gamma: F \to [0,1]$ with $\gamma(f_v) = \beta(v)$ it holds that
  $B(\gamma) = \{v \mid f_v \in B_v(\beta)\}$.
  \end{itemize}
\end{proposition}

For the hitting set, a more specific version of our main result is already known.
This result is due to Zolt{\'a}n F{\"u}redi~\cite{furedi1988}, who 
extended earlier results by Chung et al.~\cite{chung1988}. 
Recall that a hypergraph $H$ with rank $r$ is
also a $(1,r)$-hypergraph, i.e., this can be considered a special case of our setting. Furthermore, note that the statement holds only for weight minimal hitting sets.

\begin{proposition}[\cite{furedi1988}, page 152, \ Proposition 5.11.(iii)]\label{prop:furedi}
For every hypergraph $H$ of rank (i.e., maximal edge size) $r$, and every fractional hitting set $w$ for $H$  satisfying  $\weight(w)= \tau^*(H)$, the 
property $|\vsupp(w)| \leq r\cdot \tau^*(H)$ holds.
\end{proposition}

In the following we will extend Theorem~\ref{maintheor} to an analogous
statement for fractional hitting sets thereby generalizing the
previous proposition significantly. To derive the result we need a
final observation about $(c,d)$-hypergraphs. In a sense, we show that bounded multi-intersection is its own dual property.

\begin{lemma}
  \label{lem:dualbmip}
  Let $H$ be a $(c,d)$-hypergraph. Then the dual hypergraph $H^d$ is a $(d+1,c-1)$-hypergraph.\footnote{Note that the superscript of $H^d$  only signifies that it is the dual of $H$. It is not connected to the integer constant $d$ used for the multi-intersection width of $H$.}
\end{lemma}
\begin{proof}
  Let $G=(V \cup E,A)$ be the incidence graph of $H$.
  $H$ being a $(c,d)$-hypergraph is equivalent to $G$ not having $K_{c,d+1}$ as a subgraph, with
  $c$ vertices taken from $E$ and $d+1$ vertices taken from $V$.
  As the incidence graph $G^d=(W \cup F,B)$ of $H^d$ is isomorphic to $G$, with vertices and edges changing sides,
  we conclude that $G^d$ does not have $K_{d+1,c}$ as a subgraph with $d+1$ vertices taken from $F$ and
  $c$ vertices taken from $W$. This is equivalent to saying that $H^d$ is a $(d+1,c-1)$-hypergraph. 
\end{proof}

\begin{theorem}
  \label{thm:vcsupp}
  There is a function $h(c,d,k)$ such that the following is true.
  Let $c,d,k$ be constants.
  Let $H$ be a $(c,d)$-hypergraph and $\beta$ be an assignment of weights
  to $V(H)$. Assume that $weight(\beta) \leq k$.
  Then there is an assignment $\nu$ of weights to $V(H)$
  such that $weight(\nu) \leq k$, $B_v(\beta) \subseteq B_v(\nu)$
  and $|\vsupp(\nu)| \leq h(c,d,k)$.
\end{theorem}
\begin{proof}
  Let $\gamma$ be the dual weight assignment of $\beta$ as in Proposition~\ref{prop:dualcover}. That is, $\gamma : F \to [0,1]$ is an
  edge weight assignment in the dual hypergraph $H^d = (W,F)$ with $|\supp(\gamma)|=|\vsupp(\beta)|$ and $\weight(\gamma) = \weight(\beta)$.

  From Lemma~\ref{lem:dualbmip} we have that $H^d$ is a
  $(d+1,c-1)$-hypergraph and thus by Theorem~\ref{maintheor} there is an
  edge weight function $\nu'$ with $B(\gamma) \subseteq B(\nu')$ and
  $|\supp(\nu')| \leq h'(d+1,c-1,k)$.
  Let $\nu$ now be the dual weight assignment of $\nu'$. By Proposition~\ref{prop:dualcover} we then see that also $B_v(\beta) \subseteq B_v(\nu)$ and $|\vsupp(\nu)|=|\supp(\nu')|\leq h'(d+1,c-1,k)$.
\end{proof}

\section{Conclusion \& Open Questions}
\label{sec:conclusion}
\subsection{Conclusion}
We have proved novel upper bounds on the size of the support of fractional edge covers and vertex covers. 
These bounds have then been fruitfully applied to the problem of checking 
$\fhw(H) \leq k$ for given hypergraph $H$. Recall that, without imposing any restrictions on the hypergraph $H$, 
this problem is NP-complete even for $k=2$~\cite{DBLP:conf/pods/FischlGP18},
thus ruling out even XP-membership.
In contrast, for hypergraph classes that exhibit bounded multi-intersection, we have managed to establish 
XP-membership, that is, checking $\fhw(H) \leq k$ for hypergraphs in such a class is feasible in polynomial time 
for any constant $k$. Beyond the application to checking fractional hypertree width, our main result reveals completely new and far-reaching connections between fractional covers in hypergraphs and hypergraph structure which may be of independent interest in a wide variety of fields.

Below we identify a number of interesting open problem that are closely related to our main results.

\subsection{Precise Computation of $\fhw(H)$}
We have shown that for any $(c,d)$-hypergraph $H$,
the question $\fhw(H) \stackrel{?}{\leq}k$ can be answered
in polynomial time with the degree of the polynomial depending 
on $c$, $d$, and $k$. Suppose we are given a constant 
$k$ such that $\fhw(H) \leq k$.
Is it possible to compute the optimal (precise) value of
$\fhw(H)$ in polynomial time  with the degree of the 
polynomial depending on $c$, $d$, and $k$?

We know that $1 \leq \fhw(H) \leq k$ 
and that for each $1 \leq k' \leq k$ we can test
$\fhw(H) \leq k'$ in time polynomial in $c,d,k'$.
It might seem that $\fhw(k)$ can be efficiently
computed by repeated binary-search like querying
$\fhw(H) \leq k'$ for values of $k'$ getting closer and closer 
to the actual value of $\fhw(H)$.  

Unfortunately, this method
does not work. More specifically,
there is no function $h^*(c,d,k)$ upper bounding
the degree of the polynomial for the runtime of the 
resulting algorithm. Indeed, if such a function 
existed then it would hold that $h(c,d,k') \leq h^*(c,d,k)$
for any $1 \leq k'  \leq k$.
The proposition below demonstrates that this is not
the case. 

\begin{proposition} \label{incont}
$h(2,1,x)$ tends to infinity as $x$ approaches $2$ from below. 
\end{proposition}

\begin{proof}
  Recall the hypergraph family $(H_r)_{r \geq 2}$ from Example~\ref{ex:LongEdge} that we defined as follows.
  $H_r=(V_r,E_r)$ with $V_r=\{v_0, \dots, v_r\}$
  and $E=\{e_0, \dots, e_r\}$ with $e_0=\{v_1, \dots, v_r\}$
  and $e_i=\{v_0,v_i\}$ for $1 \leq i \leq r$.

  It is known that the size of the smallest fractional edge 
  cover of $H_r$ is $2-1/r$ and the cover is witnessed 
  by the unique assignment of weights where the weight
  of $e_0$ is $1-1/r$ and the weight of the rest of the hyperedges
  is $1/r$. Clearly the support of this assignment of 
  weights is $r+1$ and hence $h(2,1,2-1/r) \geq r+1$ for each
  integer $r \geq 2$.
  In fact, $H_r$ witnesses that for any
  $2-1/r \leq x<2$, $h(2,1,x) \geq r+1$. Indeed, a direct inspection shows
  that any edge cover of $H_r$ of a support of size smaller than $r+1$
  needs to have weight at least $2$. It follows that even if we set $x$
  larger than $2-1/r$ but still smaller than $2$, the support of
  size $r+1$ is needed anyway. The rest of the proof is an elementary
  calculus exercise. 
\end{proof}

 \nop{
Proposition \ref{incont} also has an interesting consequence in
the context of the dual problem of the fractional hitting set. 
Proposition \ref{prop:furedi} upper bounds the support as a function of the rank
and $\lceil k \rceil$. On the other hand Proposition \ref{incont} shows
that already for $(2,1)$-hypergraphs the support size cannot be upper bounded
as a function of $\lceil k \rceil$. Thus this nice property is lost because of the
generalization from rank to multi-intersection. }

The impossibility to efficiently compute $\fhw(H)$ for $(c,d)$-hypegraphs by the method as above, of course, does not mean that the parameter cannot be efficiently computed.
We leave the possibility of such a computation as an interesting open question.

\begin{question}  \label{polyprec}
Let $H$ be a $(c,d)$ hypergraph such that $\fhw(H) \leq k$ for some integral constant $k$.
Is it possible to compute $\fhw(H)$ in a polynomial time with the degree of the polynomial
depending on $c,d,k$?
\end{question}

It seems that a positive answer to Question \ref{polyprec} requires a new algorithmic
approach for the computation of fractional hypertree decompositions of small width where bags
do not necessarily have bounded support. This will require a deeper insight
into the structure of hypertree decompositions of hypergraphs. 

\subsection{From Bounded Multi-Intersection to Bounded VC Dimension}
It is known that $(c,d)$-hypergraphs have VC dimension at most $c+d$ \cite{JACM2021}.
Therefore, it is natural to ask whether it is possible to generalize Theorem \ref{maintheor} 
from bounded multi-intersection to bounded VC dimension.
More precisely, is there a function $f$  such that for any constants $d$ and $k$,
any hypergraph $H$
of VC dimension at most $d$ and fractional edge cover of weight at most $k$,
has a fractional edge cover $\gamma$ of weight at most $k$ and the support of $\gamma$ is of size at most $f(d,k)$?
We conjecture that the answer to this question is negative and that there is a class of hypergraphs
witnessing  the negative answer.

\begin{conjecture}
There are constants $d,k$ and an infinite class $\mathcal{H}$ of hypergraphs
whose VC dimension is at most $d$, the fractional edge cover is at most $k$
and the set $\{minsupport_k(H)|H \in \mathcal{H}\}$ is unbounded where 
$minsupport_k(H)$ is the smallest size of support of an edge cover of $H$
of weight at most $k$.  
\end{conjecture}

Let us discuss the reason why we stated the above conjecture.
A notable result \cite{DBLP:journals/talg/PhilipRS12}  in the area of parameterized complexity 
implies that the (non-fractional) set cover problem is FPT for $(c,d)$-hypergraphs
parameterized by the size $k$.
(The result is stated for dominating sets but can be reformulated in
terms of set covers through a minor modification.)
On the other hand, the problem becomes W[1]-hard already
for VC dimension $2$ \cite{HittingSetHardness}. 
Thus the set cover problem for bounded VC dimension is notably harder than for bounded
multi-intersection.

\subsection{Fixed-Parameter Tractability}
Recall that fractional hypertree width ($\fhw$) is defined as the 
$f$-width where $f$ is the {\em fractional\/} edge cover number of the 
bags of a tree decomposition. Analogously, the generalized 
hypertree width ($\ghw$) is defined as the 
$f$-width where $f$ is the {\em integral\/} edge cover number of the 
bags of a tree decomposition.
The computation of both, $\fhw$ and $\ghw$ is
hard for hypergraphs in general \cite{2009gottlob,DBLP:conf/pods/FischlGP18}.
However, our recent results demonstrate that
the generally intractable problems for the computation
of these notions of width admit XP-algorithms for restricted
classes of hypergraphs. It is therefore natural to ask whether
even more efficient algorithms, and in particular FPT-algorithms, are possible.

We believe that the $(2,1)$-hypergraphs are the right
class to start this investigation with. Even more specifically, we propose to first look at the situation for generalized hypertree width. The parameterized
intractability, if established  for this class of hypergraphs, will extend
to parameterized intractability for $(c,d)$-hypergraphs
in general. Moreover, the methods used in Section~\ref{sec:fhw} to show tractability of checking $\fhw$ rely on the tractability of checking $\ghw$ for the respective fragments. On the other hand, if an FPT-algorithm for  generalized
hypertree width is obtained, it is likely to be based on
a novel insight, thus inspiring further research
regarding a possibility of its generalization.
We believe that the case of $(2,1)$-hypergraphs is a critical starting point for such considerations
as the general case for $(c,d)$-hypergraphs may involve significant additional combinatorial challenges which are not directly relevant for the key observations. 
We therefore propose the following future research question.
\begin{question}
Is there an FPT-algorithm parameterized by $k$ that
tests $\ghw(H) \leq k$ for $(2,1)$-hypergraphs?
\end{question}

Recent work has shown that, analogously to treewidth in graphs, $\ghw$ can be characterised in terms of forbidden substructures in degree-2 hypergraphs (i.e., $(3,0)$-hypergraphs)~\cite{DBLP:journals/corr/abs-2111-11532}. Such a characterisation can provide an alternative path towards fixed-parameter tractable checking of $\ghw \leq k$ (in the degree 2 case) through deciding whether certain substructures (whose size depends on $k$) are present in the hypergraph.

\section*{Acknowledgements }
The authors acknowledge support by the Vienna Science and Technology Fund (WWTF) [10.47379/ICT2201] and the Austrian Science Fund (FWF): Project P30930.
Georg Gottlob is a Royal Society Research Professor and acknowledges support by the Royal Society in this role through the  “RAISON DATA” project  (Reference No. RP\textbackslash{}R1\textbackslash{}201074). Matthias Lanzinger acknowledges support by the Royal Society  “RAISON DATA” project  (Reference No. RP\textbackslash{}R1\textbackslash{}201074).

\bibliographystyle{acm}
\bibliography{jctb21}

\end{document}